\documentclass[journal,comsoc]{IEEEtran}

\usepackage[T1]{fontenc}

\usepackage{amssymb}
\usepackage{amsmath}
\usepackage{amsfonts}
\usepackage{graphicx}
\usepackage{subfigure}
\usepackage{cite}
\usepackage{enumerate}
\usepackage{url}
\usepackage{epstopdf}
\usepackage{algorithm}
\usepackage{algorithmic}
\usepackage{stfloats}

\newtheorem{Prop}{Proposition}

\begin{document}

\title{{Listen-and-Talk: Protocol Design and Analysis for Full-duplex Cognitive Radio Networks}}

\author{
Yun~Liao,~\IEEEmembership{Student~Member,~IEEE}, Tianyu~Wang,~\IEEEmembership{Student~Member,~IEEE}, Lingyang~Song,~\IEEEmembership{Senior~Member,~IEEE}, and Zhu~Han~\IEEEmembership{Fellow, IEEE}
\thanks{
This manuscript has been accepted by IEEE Transactions on Vehicular Technology.
}
\IEEEcompsocitemizethanks{
\IEEEcompsocthanksitem Y.~Liao, T.~Wang, and L.~Song are with the State Key Laboratory of Advanced Optical Communication Systems and Networks, School of Electronics Engineering and Computer Science, Peking University, Beijing, China (email: \{yun.liao, tianyu.alex.wang, lingyang.song\}@pku.edu.cn).
\IEEEcompsocthanksitem Z. Han is with Electrical and Computer Engineering Department, University of Houston, Houston, TX, USA (email: zhan2@uh.edu).
}
}

%
%

\maketitle
\vspace{-3em}
\begin{abstract}
In traditional cognitive radio networks, secondary users~(SUs) typically access the spectrum of primary users~(PUs) by a two-stage ``listen-before-talk''~(LBT) protocol, i.e., SUs sense the spectrum holes in the first stage before transmitting in the second. However, there exist two major problems: 1) transmission time reduction due to sensing, and 2) sensing accuracy impairment due to data transmission. In this paper, we propose a ``listen-and-talk''~(LAT) protocol with the help of full-duplex~(FD) technique that allows SUs to simultaneously sense and access the vacant spectrum. Spectrum utilization performance is carefully analyzed, with the closed-form spectrum waste ratio and collision ratio with the PU provided. Also, regarding the secondary throughput, we report the existence of a tradeoff between the secondary transmit power and throughput. Based on the power-throughput tradeoff, we derive the analytical local optimal transmit power for SUs to achieve both high throughput and satisfying sensing accuracy. Numerical results are given to verify the proposed protocol and the theoretical results.
\end{abstract}


\section{Introduction}%

With the fast development of wireless communication, spectrum resources have become increasingly scarce, motivating the development of technologies such as D2D communications~\cite{nishiyama2014relay} to improve spectrum utilization of the crowded spectrum bands. Meanwhile, as an early study by Federal Communications Commission suggests, some of the allocated spectrum is largely under-utilized in vast temporal and geographic dimensions~\cite{FCC2002report, SSC2010report}. Cognitive radio, focusing on these bands, has attracted wide attentions over the past years as a promising solution to the spectrum reuse~\cite{mitola1999cognitive, mitola2000cognitive}. In cognitive radio networks~(CRNs), unlicensed secondary users~(SUs) are allowed to access spectrum bands of the licensed primary users~(PUs) by two spectrum sharing approaches: underlay and overlay\cite{akyildiz2006next}. In underlay spectrum sharing, the SUs are allowed to operate if the interference caused to PUs is below a given threshold with proper resource management~\cite{zhou2015game, goldsmith2009breaking}. Overlay spectrum sharing, which is adopted in this paper, refers to the spectrum utilize technique that allows the SUs to access only the empty spectrum for PUs\cite{sankaranarayanan2005bandwidth}. Thus, reliable identification of spectrum holes is required to protect the PUs and maximize SUs' throughput\cite{kim2008efficient}.

\vspace{-0.5em}
\subsection{Conventional Cognitive Radio Protocols}
Most existing works on overlay CRNs employ ``listen-before-talk''~(LBT) protocol on half-duplex~(HD) radio, in which the traffic of SUs is time-slotted, and each slot is divided into two sub-slots, namely sensing sub-slot and transmission sub-slot. SUs sense the target channel in the sensing sub-slot to decide whether to access the spectrum in the following transmission sub-slot \cite{zhao2007decentralized, yucek2009survey, zhao2007optimal, liang2008sensing, huang2008short, huang2009optimal, mishra2006cooperative, cai2014optimal, gao2012security}. In \cite{zhao2007optimal, liang2008sensing, huang2008short}, optimization of sensing and transmission duration has been discussed. In \cite{huang2009optimal}, the authors considered general PU idle time distributions and imperfect sensing, and provide a tight upper bound of the performance in the LBT protocol. Cooperative spectrum sensing has been studied in \cite{mishra2006cooperative, cai2014optimal, gao2012security} to achieve better sensing performance. Though the conventional HD based LBT protocol is proved to be effective, it actually dissipates the precious resources by employing time-division duplexing, and thus, unavoidably suffers from two major problems as follow.
\begin{enumerate}
\item The SUs have to sacrifice the transmit time for spectrum sensing, and even if the spectrum hole is long and continuous, the data transmission need to be split into small discontinuous slots;
\item During the transmission sub-slots, the SUs do not sense the spectrum. Thus, if the PUs arrive or leave during the transmission sub-slots, SUs cannot be aware until the next sensing sub-slot, which leads to long collision (when the PUs arrive) and spectrum waste (when the PUs leave).
\end{enumerate}

\vspace{-0.5em}
\subsection{Utilizing Full-duplex Technique in CRNs}
A more efficient way to utilize the spectrum holes and protect the PU network should allow the SUs to keep sensing the spectrum all the time. Whenever a spectrum hole is detected, SUs begin transmission, and once the PU arrives, the transmission ceases. This can be facilitated by full-duplex~(FD) techniques~\cite{bharadia2013full}. In a FD system, a node can transmit and receive using the same time and frequency resources. However, due to the close proximity of a given modem's transmitting antennas to its receiving antennas, strong self-interference introduced by its own transmission makes decoding process nearly impossible, which had been a huge impediment to the development of FD communications in the past. Recently, there has been significant progress in the self-interference cancelation including proper hardware design and signal process techniques, presenting great potential for realizing the FD communications for the future wireless networks~\cite{bharadia2013full, kiessling2004mutual, jain2011practical, choi2010achieving}.

Motivated by the FD techniques, in this paper, we propose a ``listen-and-talk''~(LAT) protocol, by which SUs can simultaneously perform spectrum sensing and data transmission \cite{mypaper}. We assume that the PU can change its state at any time and each SU has two antennas working in FD mode. Specifically, at each moment, one of the antennas at each SU senses the target spectrum band, and judges if the PU is busy or idle; while the other antenna transmits data simultaneously or keeps silent on basis of the sensing results.

\vspace{-0.5em}
\subsection{Related Works}
We remark that the ideas of this kind have also been mentioned by some other recent works\cite{cheng2011full, afifi2013exploiting, ahmed2012simultaneous, riihonen2014energy, zheng2013full, choi2012beyond}. Some of the papers such as \cite{choi2012beyond} have mentioned the simultaneous sensing and transmission briefly as a feasible application scenario of FD technology without further analysis, while some other papers study the similar topics. We concisely summarize them as follow.

Some works focus on deploying FD radios on CR users and the impact of some physical issues leading to residual self-interference and imperfect sensing \cite{ahmed2012simultaneous, riihonen2014energy, cheng2011full}. Specifically, \cite{ahmed2012simultaneous} discussed the use of directional antennas, and showed that directionality of multi-reconfigurable antennas could increase both the range and rate of full-duplex transmissions over omni-directional antenna based full-duplex transmissions; \cite{riihonen2014energy} focused on comparing sensing error probabilities in the half-duplex, two-antenna full-duplex, and single-antenna full-duplex cognitive scenarios under energy detection; and the impact of some physical issues leading to residual self-interference and imperfect sensing such as bandwidth, antenna placement error, and transmit signal amplitude difference was discussed in \cite{cheng2011full}

In \cite{afifi2013exploiting}, the authors considered multiple SU links with partial/complete self-interference suppression capability, and they could operate in either simultaneous transmit-and-sense~(TS) or simultaneous transmit-and-receive~(TR) modes. Mode selection between the TS and TR mode and the coordination of SU links were proposed to achieve high secondary throughput. The idea of the TS mode is similar to our protocol. However, in \cite{afifi2013exploiting}, one fixed threshold for energy detection was used in both SO and TS modes, while in our work, a pair of sensing thresholds are designed to compensate for the imperfect self-interference cancelation. Besides, the authors in \cite{afifi2013exploiting} only provided calculation of error sensing probabilities in series expressions in the analysis, and failed to present how well can the SUs utilize the spectrum holes, which is addressed in our work.

The authors in \cite{zheng2013full} considered cooperation between primary and secondary systems. In their model, the cognitive base station~(CBS) relays the primary signal, and in return it can transmit its own cognitive signal. The CBS was assumed to be FD enabled with multiple antennas. Beamforming technique was used to differentiate the forwarding signal for primary users and secondary transmission.

Different from all the above works, throughout the paper, we focus on the following important issues:
\begin{itemize}
\item How to design the sensing strategy so that the benefits of the FD can be fully enjoyed?
\item How to design the secondary transmit parameter, e.g., transmit power, so that SUs can achieve high throughput as well as satisfactory sensing performance?
\item How well can the proposed LAT perform in terms of spectrum utilization efficiency and secondary throughput?
\end{itemize}

We explore answers to these questions by both theoretical analysis and simulation results. The main contributions of this paper can be summarized below.

\begin{itemize}
\item We clearly present the idea of simultaneous sensing and transmission, and design a ``listen-and-talk'' protocol indicating when and with what power should a SU access the spectrum, and how to set the detection threshold.

\item We present theoretical analysis of the sensing performance and the spectrum utilization. Especially, the closed-form expressions of the collision ratio at the primary network and the spectrum waste ratio are provided.

\item We report a power-throughput tradeoff, show the existence of a local optimal transmit power, with which the SUs can achieve high throughput as well as satisfying sensing performance, and derive the theoretical expression of the local optimal transmit power.
\end{itemize}

The rest of the paper is organized as follows. Section II describes the system model and the concept of simultaneous sensing and transmission. In Section III, we elaborate the proposed LAT protocol and discuss the design of important parameters. In Section IV, we investigate the analytical performance, including the spectrum utilization efficiency and secondary throughput. Also, a power-throughput tradeoff is reported and analyzed. Simulation results are presented to verify the analytical results in Section V. We conclude the paper in Section VI.

\section{System Model}%

In this section, the system model of the overall network is presented, and the concept of simultaneous sensing and transmission under imperfect self-interference suppression is elaborated.

\vspace{-0.5em}
\subsection{System model}

\begin{figure*}[!t]
\vspace{-1em}
\centering
\includegraphics[width=6.8in]{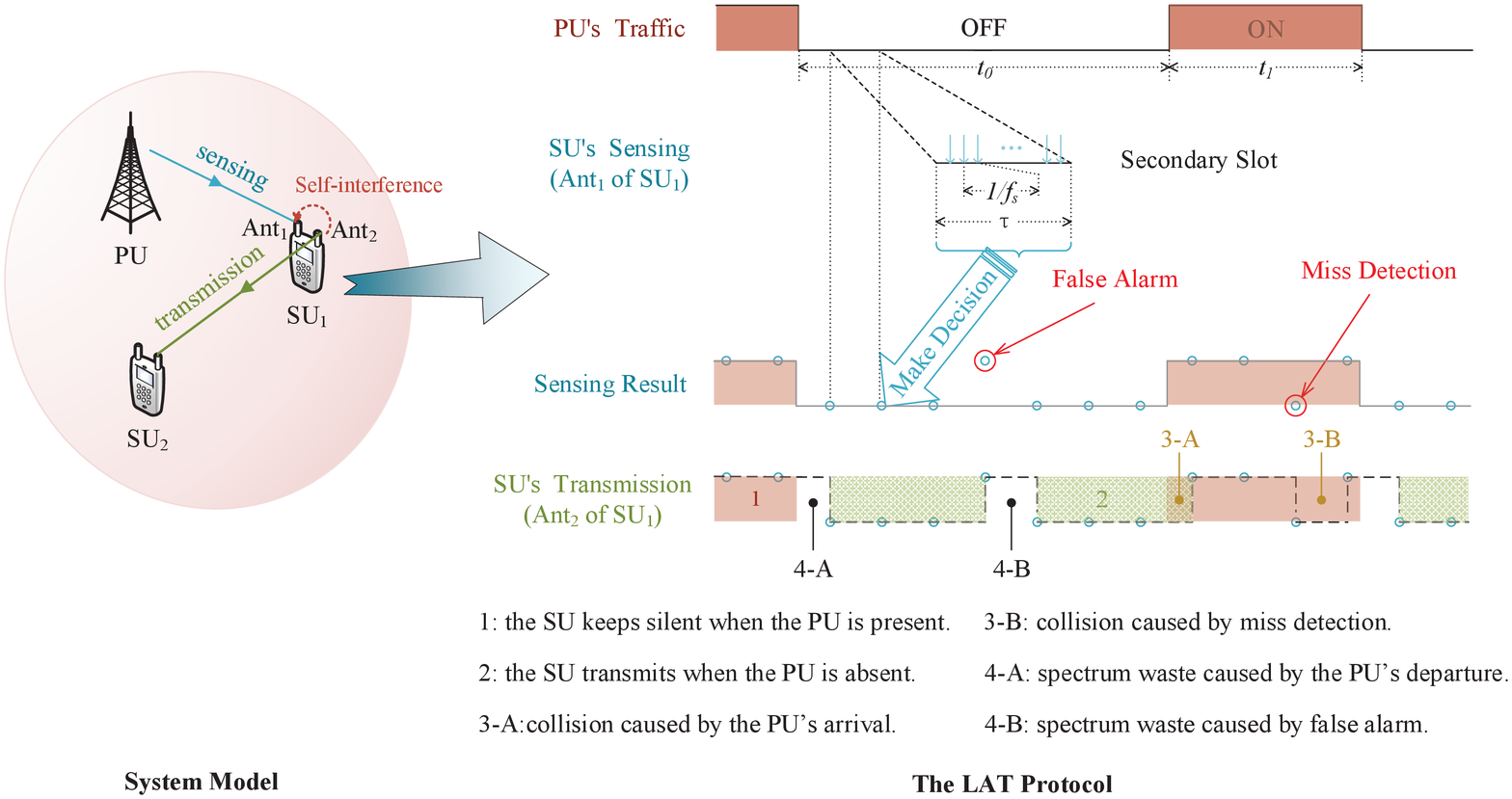}
\vspace{-1.5em}
\caption{System model: Listen-and-Talk.} \label{model}
\vspace{-1em}
\end{figure*}

We consider a CR system consisting of one PU and one SU pair as shown in the left part of Fig.~\ref{model}, in which SU$_1$ is the secondary transmitter and SU$_2$ is the receiver. Each SU is equipped with two antennas Ant$_1$ and Ant$_2$, where Ant$_1$ is used for data reception, while Ant$_2$ is used for data transmission. The transmitter SU$_1$ uses both Ant$_1$ and Ant$_2$ for simultaneous spectrum sensing and secondary transmission with the help of FD techniques, while the receiver SU$_2$ uses only Ant$_1$ to receive signal from SU$_1$.\footnote{In this model, only the secondary transmitter~(SU$_1$) performs spectrum sensing, while the receiver~(SU$_2$) does not. This transmitter-only sensing mechanism is widely adopted in today's cognitive radio, considered that secondary transmitters and receivers cannot continuously exchange sensing results without interfering the primary network. Besides, we assume that SU$_2$ has two antennas for fairness and generality, since SU$_2$ does not always need to be the receiver.}

The spectrum band occupancy by the PU is modeled as an alternating busy/idle random process where the PU can access the spectrum at any time. We assume that the probabilities of the PU's arrival and departure remain the same across the time, and the holding time of either state is distributed as the exponential distribution\cite{huang2008opportunistic}. We denote the variables of the idle period and busy period of the PU as $t_0$ and $t_1$ , respectively. And let $\tau_0 = \mathbb{E}\left[t_0\right]$ and $\tau_1 = \mathbb{E}\left[t_1\right]$ represent the average idle and busy duration. According to the property of exponential distribution, the probability density functions~(PDFs) of $t_0$ and $t_1$ can be written as, respectively,
\begin{equation}
\begin{split}
&f_0(t_0) = \frac{1}{\tau_0}e^{-\frac{t_0}{\tau_0}},\\
&f_1(t_1) = \frac{1}{\tau_1}e^{-\frac{t_1}{\tau_1}}.
\end{split}
\end{equation}

For SUs, on the other hand, only the idle period of the spectrum band is allowed to be utilized. To detect the spectrum holes and avoid collision with the PU, SU$_1$ needs to sample the spectrum at sampling frequency $f_s$, and make decisions of whether the PU is present after every $N_s$ samples, This makes the secondary traffic time-slotted, with slot length $T = N_s/f_s$.

Considering the common case that $f_s$ can be sufficiently high and the state of PU changes sufficiently slowly, we assume that $\tau_0, \tau_1 \gg T$ and $N_s$ is a sufficiently large integer. If we divide the PU traffic into slots in accordance with SU's sensing process, the probability that PU changes its state in a stochastic slot can be derived as follows.
\begin{itemize}
\item The PU arrives in a stochastic slot:
\begin{equation}\label{arrive}
\mu = \int\limits_{0}^{T} {{f_0}\left( {{t_0}} \right)d{t_0}}  = 1 - {e^{ - \frac{1}{{{m_0}}}}},
\end{equation}
where $m_0 = {\tau_0}/{T}$ and we assume it to be a large integer.
\item The PU leaves in a stochastic slot:
\begin{equation}\label{depart}
\nu = \int\limits_{0}^{T} {{f_1}\left( {{t_1}} \right)d{t_1}}  = 1 - {e^{ - \frac{1}{{{m_1}}}}},
\end{equation}
where $m_1 = {\tau_1}/{T}$ is assumed to be a large integer.
\end{itemize}
Note that when $m_0$ and $m_1$ are sufficiently large, we have $\mu\approx\frac{1}{m_0}$ and $\nu\approx\frac{1}{m_1}$.

\vspace{-0.5em}
\subsection{Simultaneous Sensing and Transmission}

With the help of FD technique, SU$_1$ can detect the PU's presence when it is transmitting signal to SU$_2$. However, as shown by the dotted arrow in the system model in Fig.~1, the challenge of using FD technique is that the transmit signal at Ant$_2$ is received by Ant$_1$, which causes self-interference at Ant$_1$. Note that for Ant$_1$, the received signal is affected by the state of the transmit antenna (Ant$_2$): when Ant$_2$ is silent, the received signal at Ant$_1$ is free of self-interference, and the spectrum sensing is the same as the conventional half-duplex sensing. Thus, we consider the circumstances when SU$_1$ is transmitting or silent
separately.

When SU$_1$ is silent, the received signal at Ant$_1$ is the combination of potential PU's signal and noise. The cases when the PU is busy or idle are referred to as hypothesis $\mathcal{H}_{01}$ and $\mathcal{H}_{00}$, respectively. The received signal at Ant$_1$ under each hypothesis can be written as
\begin{equation}\label{H0}
y = \left\{
\begin{aligned}
& h_ss_p + u,& {\mathcal{H}_{01}},\\
& u,& {\mathcal{H}}_{00},\\
\end{aligned}
\right.
\end{equation}
where $s_p$ denotes the signal of the PU, $h_s$ is the channel from the PU to Ant$_1$ of SU$_1$, and $u\sim \mathcal{CN}\left(0,\sigma_u^2\right)$ denotes the complex-valued Gaussian noise. We assume that $s_p$ is PSK modulated with variance $\sigma_p^2$, and $h_s$ is a Rayleigh channel with zero mean and variance $\sigma_h^2$.

When SU$_1$ is transmitting to SU$_2$, RSI is introduced to the received signal at Ant$_1$. The received signal can be written as
\begin{equation}\label{H1}
y = \left\{
\begin{aligned}
& h_ss_p + w + u,& {\mathcal{H}_{11}},\\
& w + u,& {\mathcal{H}}_{10},\\
\end{aligned}
\right.
\end{equation}
where $\mathcal{H}_{11}$ and $\mathcal{H}_{10}$ are the hypothesises under which the SU is transmitting and the PU is busy or idle, respectively. $w$ in \eqref{H1} denotes the RSI at Ant$_1$, which can be modeled as the Rayleigh distribution with zero mean and variance $\chi^2\sigma_s^2$\cite{jain2011practical,everett2013passive}, where $\sigma_s^2$ denotes the secondary transmit power at Ant$_2$ and $\chi^2 := \frac{\text{Power of the RSI}}{\text{Transmit power}}$ represents the degree of self-interference suppression, which is commonly expressed in dBs, and indicates how well can the self-interference be suppressed.

Spectrum sensing refers to the hypothesis test in either \eqref{H0} or \eqref{H1}. Given that SU$_1$ has the information of its own state (silent or transmitting), it can automatically choose one pair of the hypothesises to test $\left(\{\mathcal{H}_{00}, \mathcal{H}_{01}\}~\text{or}~ \{\mathcal{H}_{10}, \mathcal{H}_{11}\}\right)$ and decide whether the PU is present or not.

\section{Listen-and-Talk~(LAT) Protocol and Key Parameter Design}

In this section, we first present the proposed LAT protocol, and then discuss the key parameter design in spectrum sensing to meet the constraint of collision ratio to the primary network.

\vspace{-0.5em}
\subsection{The LAT Protocol}
The right part of Fig.~\ref{model} shows the LAT protocol, in which SU$_1$ performs sensing and transmission simultaneously by using the FD technique: Ant$_1$ senses the spectrum continuously while Ant$_2$ transmits data when a spectrum hole is detected. Specifically, SU$_1$ keeps sensing the spectrum with Ant$_1$ with sampling frequency $f_s$, which is shown in the line with down arrows. At the end of each slot with duration $T$, SU$_1$ combines all samples in the slot and makes the decision of the PU's presence. The decisions are represented by the small circles, in which the higher ones denote that the PU is judged active, while the lower ones denote otherwise. The activity of SU$_1$ is instructed by the sensing decisions, i.e., SU$_1$ can access the spectrum in the following slot when the PU is judged absent, and it needs to backoff otherwise.

Since the PU can change its state freely, there exist the following four states of spectrum utilization:
\begin{itemize}
    \item State$_1$: the spectrum is occupied only by the PU, and SU$_1$ is silent.
    \item State$_2$: the PU is absent, and SU$_1$ utilizes the spectrum.
    \item State$_3$: the PU and SU$_1$ both transmit, and a collision happens.
    \item State$_4$: neither the PU nor SU$_1$ is active, and there remains a spectrum hole.
\end{itemize}

Among these four states, State$_1$ and State$_2$ are the normal cases, and State$_3$ and State$_4$ are referred to as collision and spectrum waste, respectively. There are two reasons leading to State$_3$ and State$_4$: (A) the PU changes its state during a slot, and (B) sensing error, i.e., false alarm and miss detection.

\vspace{-0.5em}
\subsection{Energy Detection}
We adopt energy detection as the sensing scheme, in which the average received power in a slot is used as the test statistics $M$:
\begin{equation}\label{test}
M = \frac{1}{{{N_s}}}\sum\limits_{n = 1}^{{N_s}} {{{\left| {y\left( n \right)} \right|}^2}},
\end{equation}
where $y\left(n\right)$ denotes the $n$th sample in a slot, and the expression for $y\left(n\right)$ is given in \eqref{H0} and \eqref{H1}.

With a chosen threshold $\epsilon$, the spectrum is judged occupied when $M \geq \epsilon$, otherwise the spectrum is idle. Generally, the probabilities of false alarm and miss detection can be defined as,
\begin{equation}
\begin{split}
&{P_f}\left( \epsilon  \right) = \Pr \left( {M > \epsilon |{\mathcal{H}_{0}}} \right),\\
&{P_m}\left( \epsilon  \right) = \Pr \left( {M < \epsilon |{\mathcal{H}_{1}}} \right),
\end{split}
\end{equation}
where $\mathcal{H}_0$ and $\mathcal{H}_1$ are the hypothesises when the PU is idle and busy, respectively.

Considering the difference of the received signal caused by RSI, we can achieve better sensing performance by changing the threshold according to SU$_1$'s activity. Let $\epsilon_0$ and $\epsilon_1$ be the thresholds when SU$_1$ is silent and busy, respectively, and we can have two sets of probabilities of false alarm and miss detection accordingly, denoted as $\{P_f^0\left(\epsilon_0\right), P_m^0\left(\epsilon_0\right)\}$ and $\{P_f^1\left(\epsilon_1\right), P_m^1\left(\epsilon_1\right)\}$, respectively.

\vspace{-0.5em}
\subsection{Key Parameter Design}
The most important constraint of the secondary networks is that their interference to the primary network must be under a certain level. In this article, we consider this constraint as the collision ratio between SUs and the PU, defined as
\begin{equation}
P_c = \lim_{t\rightarrow\infty}\frac{\text{Collision duration}}{\text{PU's transmission time during $[0,t]$}}.\nonumber
\end{equation}
The sensing parameters are designed according to the constraint of $P_c$. In the rest of this subsection, sensing performance is evaluated, based on which we provide the analytical design of the sensing thresholds.

\textbf{Sensing Error Probabilities: }With the statistical information of the received signal in \eqref{H0} and \eqref{H1}, the statistical properties of $M$ under each hypothesis can be derived. We consider the following two types of time slots:
\begin{itemize}
\item \emph{Slots in which the PU changes its state:} if the PU arrives in a certain slot, the received signal power is likely to increase in the latter fraction of the slot, and the average signal power~($M$) is likely to be higher than the previous slots when the PU is absent. Then the probability of correct detection is higher than $P_f^i$, with $i$ denoting the current activity of SU$_1$. Similarly, if the PU leaves in a slot, the probability of correct detection is higher than $P_m^i$. Note that these slots are rare in the whole traffic, we only consider the lower limits of correct detection in these slots, i.e., we set without further derivation the probabilities of correct detection to be $P_f^i$ and $P_m^i$ when the PU arrives or leaves, respectively.
\item \emph{Slots in which the PU remains either present or absent:} in these slots, the received signal $y\left(n\right)$ in the same slot is i.i.d., and as we assumed in Section II-A, the number of samples $N_s$ is sufficiently large. According to central limit theorem (CLT), the PDF of $M$ can be approximated by a Gaussian distribution $M\sim\mathcal{N}(\mathbb{E}\left[\left|y\right|^2\right],\frac{1}{N_s}\rm var\left[\left|y\right|^2\right])$. The specific statistical properties and the description under each hypothesis are given in Table~\ref{FDGaussian}, in which $\gamma_s = \frac{\sigma_p\sigma_h^2}{\sigma_u^2}$ denotes the signal-to-noise ratio~(SNR) in sensing, and $\gamma_i = \frac{\chi^2\sigma_s^2}{\sigma _u^2}$ is the interference-to-noise ratio~(INR). Detailed derivation of the distribution properties are provided in Appendix A.

\begin{table}[!t]
\renewcommand{\arraystretch}{1.3}
\caption{Properties of PDFs of LAT protocol} \label{FDGaussian}
\centering
\begin{tabular}{|c|c|c|c|c|c|}
\hline
Hypothesis & PU & SU & $\mathbb{E} \left[ M \right]$ & ${\mathop{\rm var}} \left[ M \right]$\\ \hline
$\mathcal{H}_{00}$ & idle & silent & $\sigma_u^2$ & $\frac{\sigma _u^4}{N_s}$\\ \hline
$\mathcal{H}_{01}$ & busy & silent & $\left( {1+\gamma _s} \right)\sigma _u^2$ & $\frac{\left(1+\gamma_s\right)^2\sigma_u^4}{N_s}$ \\ \hline
$\mathcal{H}_{10}$ & idle & active & $\left(1+\gamma_i\right)\sigma_u^2$ & $\frac{\left(1+\gamma_i\right)^2\sigma_u^4}{N_s}$ \\ \hline
$\mathcal{H}_{11}$ & busy & active & $\left(1+\gamma_s + \gamma_i\right)\sigma_u^2$ & $\frac{\left(1+\gamma_s+\gamma_i\right)^2\sigma_u^4}{N_s}$ \\ \hline
\end{tabular}
\end{table}
\end{itemize}

Based on Table \ref{FDGaussian}, the sensing error probabilities can be derived.
\begin{itemize}
\item When SU$_1$ is silent and the test threshold is $\epsilon_0$, the probability of miss detection~($P_m^0$) and the probability of false alarm~($P_f^0$) can be written as
    \begin{equation}\small\label{PmF0}
    P_m^0\left( \epsilon_0  \right) = 1 - \mathcal{Q}\left( {\left( {\frac{{{\epsilon _0}}}{{\left( {1 + {\gamma _s}} \right)\sigma _u^2}} - 1} \right)\sqrt {N_s} } \right),
    \end{equation}
    and
    \begin{equation}\small\label{PfF0}
    P_f^0\left( \epsilon_0  \right) = \mathcal{Q}\left( {\left( {\frac{\epsilon_0 }{{\sigma _u^2}} - 1} \right)\sqrt {N_s} } \right),
    \end{equation}
    respectively, where $\mathcal{Q}(\cdot)$ is the complementary distribution function of the standard Gaussian distribution.
\item Similarly, when SU$_1$ is transmitting with the threshold $\epsilon_1$, the miss detection probability ($P_m^1$) and the false alarm probability ($P_f^1$) are, respectively,
    \begin{equation}\small\label{PmF1}
    P_m^1\left( \epsilon_1  \right) = 1 - \mathcal{Q}\left( {\left( {\frac{\epsilon_1 }{{\left( {1 + {\gamma _s} + {\gamma _i}} \right)\sigma _u^2}} - 1} \right)\sqrt {N_s} } \right),
    \end{equation}
    and
    \begin{equation}\small\label{PfF1}
    P_f^1\left( \epsilon_1  \right) = \mathcal{Q}\left( {\left( {\frac{\epsilon_1 }{{\left( {1 + {\gamma _i}} \right)\sigma _u^2}} - 1} \right)\sqrt {N_s} } \right).
    \end{equation}
\end{itemize}

\textbf{State Transition and Overall Collision Probability:}
Different from the conventional LBT protocol in HD CRNs where each slot is independent, in the LAT protocol, the selection of sensing threshold depends on SUs' activity, which is instructed by the sensing result in the previous slot. Thus, the state of the system in each slot is no longer independent, and the collision ratio is not only related to sensing error probabilities in each slot, but also the state in the previous slots. Thus, joint consideration of the transition among all kinds of slots is necessary. Since the sensing error probabilities in the slots where the PU changes its presence can be approximated by that in the other slots, in this part, we model the state transition of the system as a discrete-time Markov chain (DTMC), in which the system can be viewed as totally time-slotted with $T$ as the slot length. Fig.~\ref{transition} shows the state transition diagram, where we denote State$_i$ as $S_{i\mod 4}$~($i = 1,2,3,4$) for simplicity.

\begin{figure}[!t]
\centering
\vspace{-1em}
\includegraphics[width=3.6in]{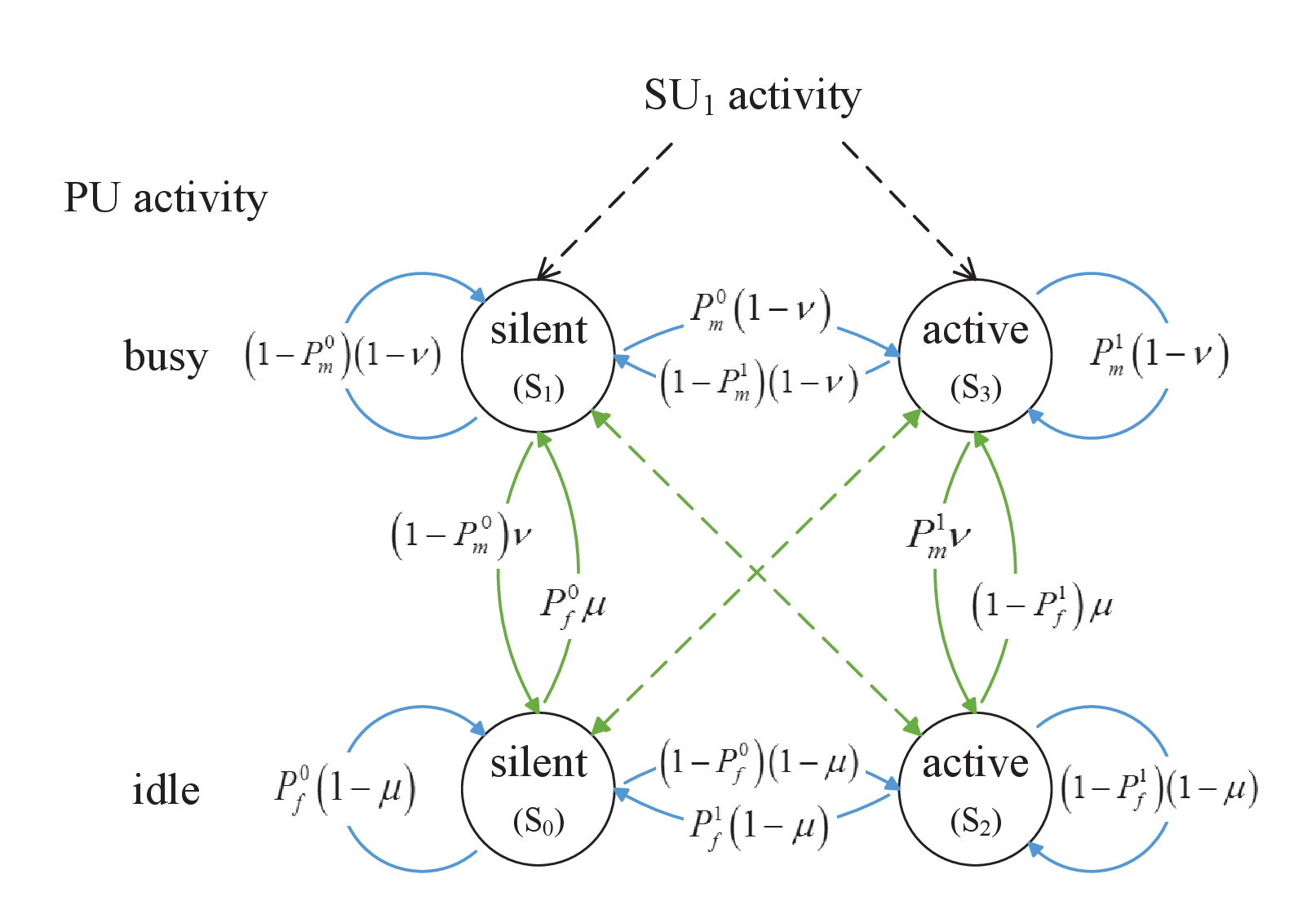}
\vspace{-1.5em}
\caption{State Transition of the System} \label{transition}
\vspace{-1em}
\end{figure}

\begin{Prop}
The probability that the system stays in the collision state $S_3$ is
\begin{equation}
P_3 = \frac{1}{{r + 1}} \cdot \frac{{P_m^0\left( {1 - \xi \Delta } \right) + \left( {1 - P_f^0} \right)r}}{{\left( {1 - \xi \Delta } \right)\varsigma  + \xi r}},
\end{equation}
where $\xi = 1 - P_f^0 + P_f^1$, $\zeta = 1 + P_m^0-P_m^1$, $r=\nu/\mu$, and $\Delta=1+r-1/\mu$.
\end{Prop}
\begin{proof}
The probability for the system staying in each state $P_k~(k=0,1,2,3)$ can be calculated considering the steady-state distribution of the Markov chain:
\begin{equation}
{\bf{\Psi p = p}},
\end{equation}
where ${\bf{p}} = \left[P_0,P_1,P_2,P_3\right]^{\bf{T}}$ is the vector of steady probabilities, and $\bf{\Psi}$ is the state transition matrix abstracted from Fig.~\ref{transition}, which is given at the top of next page.
\begin{figure*}[ht]
\begin{equation}
{\bf{\Psi}} =
\left[ {\begin{array}{*{20}{c}}
{P_f^0\left( {1 - \mu } \right)}&{P_f^0\mu \left( {1 - P_m^0} \right)\nu }&{P_f^1\left( {1 - \mu } \right)}&{\left( {1 - P_m^1} \right)\nu }\\
{P_f^0\mu }&{\left( {1 - P_m^0} \right)\left( {1 - \nu } \right)}&{P_f^1\mu }&{\left( {1 - P_m^1} \right)\left( {1 - \nu } \right)}\\
{\left( {1 - P_f^0} \right)\left( {1 - \mu } \right)}&{P_m^0\nu }&{\left( {1 - P_f^1} \right)\left( {1 - \mu } \right)}&{P_m^1\nu }\\
{\left( {1 - P_f^0} \right)\mu }&{P_m^0\left( {1 - \nu } \right)}&{\left( {1 - P_f^1} \right)\mu }&{P_m^1\left( {1 - \nu } \right)}
\end{array}} \right]\nonumber
\end{equation}
\hrulefill
\end{figure*}

Combining the constraint that $\sum\limits_{k = 0}^3 {{P_k} = 1}$, we have
\begin{equation}\small\label{Poriginal}
{\mathbf{p}} = \frac{1/(r + 1)}{{\left( {1 - \xi \Delta } \right)\varsigma  + \xi r}} \left[ {\begin{array}{*{20}{c}}
{r  \left( {P_f^1\left( {r - \varsigma \Delta } \right) + 1 - P_m^1} \right)}\\
{\left( {1 - P_m^1} \right)\left( {1 - \xi \Delta } \right) + P_f^1r}\\
{r \left( {\left( {1 - P_f^0} \right)\left( {r - \varsigma \Delta } \right) + P_m^0} \right)}\\
{P_m^0\left( {1 - \xi \Delta } \right) + \left( {1 - P_f^0} \right)r}
\end{array}} \right].
\end{equation}

To have a check on the result, we consider the probability that the PU is busy and idle as $P_{busy} = P_1+P_3 = \frac{\mu}{\mu+\nu} \approx \frac{m_1}{m_0+m_1}$ and $P_{idle} = P_0 + P_2 = \frac{\nu}{\mu+\nu}\approx \frac{m_0}{m_0+m_1}$, which are consistent with the results when we consider the PU's traffic only.
\end{proof}

\textbf{Collision Ratio:}
The collision of the PU and SU$_1$ occurs in the following two kinds of circumstances: (A) When the PU keeps occupying the spectrum and SU$_1$ fails to detect the presence of PU's signal in the previous slot, which is depicted in Fig.~\ref{transition} as $S_3$ with the probability of $P_3$. The collision length is $T$. (B) The certain slots in which PU arrives. SU$_1$ is very likely to be transmitting in these slots since the PU is likely to be absent in the previous ones. The occurrence probability of this circumstance is equal to the PU's arrival rate $\frac{\mu\nu}{\mu+\nu}$.

\begin{Prop}
The average collision length under circumstance (B), where the PU changes state, can be approximated by $\frac{T}{2}$, when $m_0$ is large enough.
\end{Prop}

\begin{proof}
and the average collision length in this case can be calculated as
\begin{equation}\small\label{col2}
\begin{split}
&\bar{T_2} = \frac{\int\limits_{0}^{T} {\left(T-t_0\right){f_0}\left( {{t_0}} \right)d{t_0}}}{\int\limits_{0}^{T} {{f_0}\left( {{t_0}} \right)d{t_0}}} = T\left(1 - m_0 - \frac{e^{-\frac{1}{m_0}}}{1-e^{-\frac{1}{m_0}}}\right) \\
 &\approx T\left(\frac{1-e^{-\frac{1}{m_0}}} {1-e^{-\frac{1}{m_0}}+\frac{1}{m_0}e^{-\frac{1}{m_0}}}\right) \approx T \left(\frac{e^{-\frac{1}{m_0}}} {e^{-\frac{1}{m_0}}+e^{-\frac{1}{m_0}}}\right) = \frac{T}{2},
\end{split}
\end{equation}
where the approximation is valid when $m_0$ is sufficiently large.
\end{proof}

It is unavoidable in the LAT protocol that when the PU arrives, a short head of the signal, with the length of a SU's slot approximately, collides with the SU's signal. Combine the two circumstances, and the overall collision rate can be given by
\begin{equation}\small\label{collision}
P_c = {\left(P_3 + \frac{1}{2}\frac{\mu\nu}{\mu+\nu}\right)}/{P_{busy}} = \frac{\nu}{2} + \frac{{P_m^0\left( {1 - \xi\Delta } \right) + \left( {1 - P_f^0} \right)r}}{{ {\left( {1 - \xi\Delta } \right)\zeta + \xi r} }},
\end{equation}

\textbf{Design of Sensing Thresholds:}
For the parameter design, we have a maximum allowable $P_c$ as the system constraint, and all the parameters of the sensing process should be adjusted according to $P_c$. Note that $\Delta$ and $r$ are only related to the PU's traffic, and $\{P_m^0,P_f^0\}$ and $\{P_m^1,P_f^1\}$ are closely related via thresholds $\epsilon_0$ and $\epsilon_1$, respectively. Thus, we actually have two independent variables of the secondary network to design to meet the constraint of $P_c$.

We choose $P_m^0$ and $P_m^1$ as the independent variables. With \eqref{collision} as the only constraint, there are infinite choices of ($P_m^0,P_m^1$) pair. For simplicity, we set $P_m^0 = P_m^1 = P_m$, i.e., $\zeta = 1$ to reduce the degree of freedom, and the constraint can be simplified as
\begin{equation}\small\label{Pc}
P_c = \frac{\nu}{2} + \frac{{P_m\left( {1 - \xi\Delta } \right) + \left( {1 - P_f^0} \right)r}}{1 + \left(\frac{1}{\mu} - 1\right)\xi},
\end{equation}
where $\Delta$, $r$, $\mu$, and $\nu$ are relevant only to the PU traffic, and $\xi = 1 - P_f^0 + P_f^1$ can be derived from $P_m$ via test thresholds $\epsilon_0$ and $\epsilon_1$.

In the rest of this part, we calculate $P_m$ from the constraint of $P_c$, from which the sensing thresholds $\epsilon_0$ and $\epsilon_1$ can be obtained.

Combining \eqref{PmF0} and \eqref{PfF0}, \eqref{PmF1} and \eqref{PfF1}, we can obtain $P_f^0$ and $P_f^1$ as functions of $P_m$ as, respectively,
\begin{equation}\small\label{PfF0final}
P_f^0\left( {{P_m}} \right) = \mathcal{Q}\left( {-{\mathcal{Q}^{ - 1}}\left( {P_m} \right)\left( {1 + {\gamma _s}} \right) + {\gamma _s}\sqrt {N_s} } \right);
\end{equation}
\begin{equation}\small\label{PfF1final}
{P_f^1}\left( P_m \right) = \mathcal{Q}\left( {-\mathcal{Q}^{ - 1}\left( {P_m} \right)\left( {1 + \frac{{{\gamma _s}}}{{1 + {\gamma _i}}}} \right) + \frac{\gamma _s\sqrt {N_s} }{{1 + {\gamma _i}}} } \right).
\end{equation}
From \eqref{PfF0final} and \eqref{PfF1final}, we can find a rise of the false alarm probability when the RSI exists. This result indicates that when interference increases, the sensing performance gets worse.

With \eqref{PfF0final} and \eqref{PfF1final}, $\xi$ can be expressed as $\xi\left(P_m\right) = 1 - P_f^0\left(P_m\right) + P_f^1\left(P_m\right)$. With given parameters of the PU's traffic and the slot length, $P_m$ can be solved from \eqref{Pc}. Since the analytical expression of $P_m$ is complicated, we only give some typical numerical solution in Fig.~\ref{numPm}, where the sensing SNR $\gamma_s = -10$dB, INR $\gamma_i = 5$dB, number of samples $N_s = 200$, and $r$ is set to be 6 to meet the real case that the typical spectrum occupancy is less than $15\%$~\cite{FCC2002report}.

\begin{figure}[!t]
\centering
\includegraphics[width=3.6in]{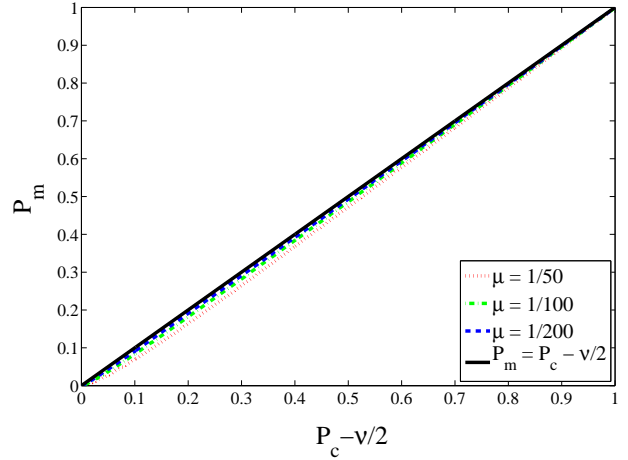}
\vspace{-1.5em}
\caption{Numerical Solution of $P_m$; $\gamma_s$ = -10dB, $\gamma_i$ = 5dB, $N_s = 200$, and $r = 5$} \label{numPm}
\vspace{-1em}
\end{figure}

It is shown in Fig.~\ref{numPm} that when $\mu$ goes down, $P_c - \frac{\nu}{2}$ becomes a fine approximation of $P_m$. With the large-$m_0$ assumption, we regard $\mu$ as sufficiently small, and $P_m$ is determined by
\begin{equation}
P_m = P_c - \frac{\nu}{2} = P_c - \frac{1}{2}\left(1-e^{-\frac{T}{\tau_1}}\right).
\end{equation}
This indicates that with the same constraint $P_c$ and parameters of the PU's traffic, the required $P_m$ gets squeezed when SU's slot length $T$ increases.

With $P_m = P_c - \nu / 2$, the thresholds $\epsilon_0$ and $\epsilon_1$ can be obtained from \eqref{PmF0} and \eqref{PmF1}, respectively:
\begin{equation}\small
\epsilon_0 = \left.\left( {\frac{\mathcal{Q}^{ - 1}\left(1 - P_m \right)}{\sqrt {N_s} } + 1} \right)\left( {1 + {\gamma _s}} \right)\sigma _u^2\right|_{P_m = P_c - \nu / 2};\label{thhold0}
\end{equation}
\begin{equation}\small
\epsilon_1 = \left.\left( {\frac{{{\mathcal{Q}^{ - 1}}\left( {{1 - P_m}} \right)}}{{\sqrt {N_s} }} + 1} \right)\left( {1 + {\gamma _s} + {\gamma _i}} \right)\sigma _u^2\right|_{P_m = P_c - \nu / 2}.\label{thhold1}
\end{equation}
A lift of sensing threshold due to the RSI~($\gamma_i$) can be found from \eqref{thhold0} and \eqref{thhold1}, which is in accordance with the previous analysis.

\section{Performance Analysis of the LAT}

In this section, we first evaluate the sensing performance of the LAT by the probabilities of spectrum waste ratio under the constraint of collision ratio. Then, with the closed-form analytical secondary throughput, a tradeoff between the secondary transmit power and throughput is elaborated theoretically.

\vspace{-0.5em}
\subsection{Spectrum Utilization Efficiency and Secondary Throughput}

\textbf{Spectrum Waste Ratio:}
Similar to the analysis of the collision ratio, we combine the following two kinds of time slots to derive the spectrum waste ratio: (A) the slots when the spectrum remains idle; and (B) the slots of the PU's departure. There exists waste of spectrum holes in (A) when the system is in the state $S_0$ in Fig.~\ref{transition}, and the probability is given by $P_0$ in \eqref{Poriginal}. Every time when the SU fails to find the hole, the waste length is $T$. In (B), the average waste length can be derived from the PU's traffic with the similar method in \eqref{col2}, and it also yields $\frac{T}{2}$ of the average waste length. The probability of the PU's departure is $\frac{\mu\nu}{\mu+\nu}$, which is same as its arrival rate. The ratio of wasted spectrum hole is then given by
\begin{equation}\small\label{Pw}
P_w = {\left(P_0 + \frac{1}{2}\frac{\mu\nu}{\mu+\nu}\right)}/{P_{idle}} = \frac{\mu}{2} + \frac{{\left( \frac{1}{\mu} - 1 \right)P_f^1 + 1 - P_m}}{1 + \left(\frac{1}{\mu} - 1\right)\xi}.
\end{equation}

\textbf{Secondary Throughput:}
SU$_1$'s throughput can be measured with the waste ratio and the transmit rate. The achievable rate under perfect sensing is given as
\begin{equation}
R = {\log _2}\left( 1+\gamma_t \right),
\end{equation}
where $\gamma_t = \frac{\sigma_s^2\sigma_t^2}{\sigma_u^2}$ represents the SNR in transmission, with $\sigma_t^2$ denotes the channel gain of the transmit channel from SU$_1$ to SU$_2$, and SU$_1$'s throughput can be measured as
\begin{equation}\small\label{thputF}
\begin{split}
C &= R\cdot\left(1-P_w\right) \\
&= {\log _2}\left( 1+\gamma_t \right)\left(1 - \frac{\mu}{2} - \frac{{\left( \frac{1}{\mu} - 1 \right)P_f^1 + 1 - P_m}}{1 + \left(\frac{1}{\mu} - 1\right)\xi}\right).
\end{split}
\end{equation}

\vspace{-0.5em}
\subsection{Power-Throughput Tradeoff}
In the expression of SU$_1$'s throughput in \eqref{thputF}, there are two factors: $R$ and $\left(1-P_w\right)$. On one hand, $R$ is positively proportional to SU$_1$'s transmit power $\sigma_s^2$. On the other hand, however, the following proposition holds.

\begin{Prop}
The spectrum waste ratio $P_w$ increases with the secondary transmit power $\sigma_s^2$.
\end{Prop}

\begin{proof}
Firstly, the INR $\gamma_i$ increases with the transmit power and in turn lifts $P_f^1$, which can be seen from \eqref{PfF1final}.

Then, we can rewrite \eqref{Pw} as
\begin{equation}\small
\begin{split}
P_w &= \frac{\mu}{2} + \frac{{\left( \frac{1}{\mu} - 1 \right)P_f^1 + 1 - P_m}}{1 + \left(\frac{1}{\mu} - 1\right)\xi}\\
 &= \frac{\mu}{2} + \frac{{\left( \frac{1}{\mu} - 1 \right)P_f^1 + 1 - P_m}}{1 + \left(\frac{1}{\mu} - 1\right)\left(1 - P_f^0 \right) + \left(\frac{1}{\mu} - 1\right)P_f^1}\\
 &= \frac{\mu}{2} + 1 - \frac{\left(\frac{1}{\mu}-1\right)\left(1-P_f^0\right)+P_m}{{1 + \left(\frac{1}{\mu} - 1\right)\left(1 - P_f^0 \right) + \left(\frac{1}{\mu} - 1\right)P_f^1}}.
\end{split}
\end{equation}
When $P_f^1$ increases, the third term decreases and $P_w$ increases monotonically. Then the increase of SU$_1$'s transmit power results in greater waste of the vacant spectrum.
\end{proof}

Thus, there may exist a power-throughput tradeoff in this protocol: when the secondary transmit power is low, the RSI is negligible, the spectrum is used more fully with small $P_w$, yet the ceiling throughput is limited by $R$; when the transmit power increases, the sensing performance get deteriorated, while at the same time, SU$_1$ can transmit more data in a single slot.

\textbf{Local Optimal Transmit Power:}\footnote{Note that the secondary throughput is not purely convex throughout the domain of transmit power. There may exist local optimal points in low power region, while the throughput is monotonically increasing in the high power region. The point of the discussion of the power-throughput tradeoff and the calculation of the local optimal transmit power is that the secondary throughput does not monotonously increase with the transmit power, which means that SUs may not always transmit with its maximum transmit power to achieve highest throughput, instead, a mediate value may lead to better performance.}
The analysis above indicates the existence of a mediate secondary transmit power to achieve both high spectrum utilization efficiency in time domain and high secondary throughput. To obtain this mediate value of transmit power, we differentiate the expression of the throughput to find the local optimal points of the secondary transmit power $\widehat{\sigma_s^2}$, which satisfies
\begin{equation}
{\left. {\frac{{dC}}{{d\sigma _s^2}}} \right|_{\widehat {\sigma _s^2}}} = 0.
\end{equation}
With detailed derivation presented in Appendix B, we have the local optimal power satisfies
\begin{equation}\small\label{dCintext}
\kappa \ln \left( {{\gamma _t} + 1} \right) \frac{{\exp \left( { - \frac{{{\rho ^2}}}{2}} \right)  \left( {\frac{1}{\mu } - 1} \right) \Xi}}{{\sqrt {2\pi } \left(\gamma_i+1\right)^2\alpha }} + \frac{\sigma _t^2 \left( {\frac{\mu }{2} - \kappa } \right)}{\gamma_t+1}  = 0,
\end{equation}
where the notations are as follow:
\begin{equation}\small
\begin{split}
&\rho = -\mathcal{Q}^{-1}\left(P_m\right)\left(\frac{\gamma_s}{\gamma_i+1}+1\right)+\frac{\gamma_s}{\gamma_i+1}\sqrt{N_s},~ \text{i.e.},~\mathcal{Q}\left(\rho\right) = P_f^1,\\
&\alpha  = \left( {\frac{1}{\mu } - 1} \right) \cdot \left( {\mathcal{Q}\left( \rho  \right) - P_f^0 + 1} \right) + 1,\\
&\kappa  = \frac{{\left( {\frac{1}{\mu } - 1} \right) \cdot \left( {1 - P_f^0} \right) + {P_m}}}{{\left( {\frac{1}{\mu } - 1} \right) \left(\mathcal{Q}\left(\rho\right)-P_f^0+1\right)  + 1}},\\
&\Xi  = {\gamma _s}{\chi ^2}\left( {{\mathcal{Q}^{ - 1}}\left( {1 - {P_m}} \right) + \sqrt {{N_s}} } \right).
\end{split}
\nonumber
\end{equation}
With $\sigma_s^2$ as the only unknown variable, it can be calculated numerically.

To obtain better comprehension about the properties of the local optimal transmit power, we consider the case when $\mu$ is sufficiently small, and \eqref{dCintext} can be simplified as
\begin{equation}\small\label{simpdC=0}
\exp \left( { - \frac{{{\rho ^2}}}{2}} \right)\frac{\left(\gamma_t + 1\right)\ln\left( {{\gamma _t} + 1} \right)}{{{\left( {{\gamma _i} + 1} \right)^2}}} = \frac{{\sqrt {2\pi } \sigma _t^2}}{\Xi }\left( {1 - P_f^0 + \mathcal{Q}\left( \rho  \right)} \right).
\end{equation}

Now we provide the existence conditions of the local optimal transmit power. The left side of \eqref{simpdC=0} is a convex curve of $\sigma_s^2$ with a single maximum. When $\sigma_s^2$ goes to zero or infinity, the value of the left side goes to zero. The value of the right side changes from $\frac{{\sqrt {2\pi } \sigma _t^2}}{\Xi }$ to $\frac{{\sqrt {2\pi } \sigma _t^2}}{\Xi }\left( {2 - P_f^0 - P_m} \right)$. We can roughly say that when the maximum of the left side is larger than either $\frac{{\sqrt {2\pi } \sigma _t^2}}{\Xi }$ or $\frac{{\sqrt {2\pi } \sigma _t^2}}{\Xi }\left( {2 - P_f^0 - P_m} \right)$, there would be two solutions to \eqref{simpdC=0}. When the maximum of the left side is smaller than the minimum of the right side, on the other hand, no solution exists.

\textbf{Characteristics of the Power-throughput Curve:} Given the discussions above, there exist two cases of the power-throughput curve regarding the existence of the local optimal power. We analyze these two cases separately in the following.

\begin{enumerate}
\item When equation \eqref{simpdC=0} has no solutions, the curves of transmit power on the left and right sides never meet. Since the right side of \eqref{simpdC=0} is always far above zero and the left can go to zero when the transmit power is extremely high or low, we can safely say that the left side is always smaller than the right, i.e.,
\begin{equation}\small
\exp \left( { - \frac{{{\rho ^2}}}{2}} \right)\frac{\left(\gamma_t + 1\right)\ln\left( {{\gamma _t} + 1} \right)}{{{\left( {{\gamma _i} + 1} \right)^2}}} < \frac{{\sqrt {2\pi } \sigma _t^2}}{\Xi }\left( {1 - P_f^0 + \mathcal{Q}\left( \rho  \right)} \right).
\end{equation}
Substituting the inequation to \eqref{dCorigin}, we have $\frac{dC}{d\sigma_s^2} > 0$, which indicates that the secondary throughput would increase with the transmit power monotonously.

\item When the solutions of \eqref{simpdC=0} exist, we discuss the sign of $\frac{dC}{d\sigma_s^2}$ piecewise. When the power is low or high enough, the left side is small, while the right remains considerable. The solid red curve~(maximum of the left side of \eqref{simpdC=0}) is below the dash-dotted blue one~(value of the right side of \eqref{simpdC=0}), and $\frac{dC}{d\sigma_s^2} > 0$. When the power is between the two solutions, we have $\frac{dC}{d\sigma_s^2} < 0$. Thus, at the smaller solution, $\frac{{{d^2}C}}{{d{{\left( {\sigma _s^2} \right)}^2}}}<0$, and this is the local optimal transmit power $\widehat{\sigma_s^2}$ to achieve local maximum throughput. Similarly, the larger solution denotes the local minimum of the throughput.
\end{enumerate}

As an example, we plot the curves of the maximum of the left side and the corresponding value of the right side in Fig.~\ref{checksolution}. It is shown that when $\chi^2$ is smaller than 0.86, the maximum of the left is larger than the corresponding value of the right, and \eqref{simpdC=0} will have solutions and power-throughput is likely to exist. When $\chi^2$ is greater than 0.85, there may be no tradeoff between the transmit power and secondary throughput, which is verified by the thick solid line in Fig.~\ref{power_thput}.

\section{Simulation Results}%

In this section, simulation results are presented to evaluate the performance of the proposed LAT protocol. Monte Carlo simulations are performed by varying channel conditions and the PU's state. We set the default values of the simulation parameters as follow: the sample number in each slot $N_s = 300$, the corresponding probability that the PU arrives in a stochastic slot $\mu = 1/500$, and the probability that the PU leaves in any slot $\nu$ as $6/500$. The constraint of collision ratio is set as 0.1, and SNR in sensing $\gamma_s$ is assumed to be -5dB.

\vspace{-0.5em}
\subsection{Power-Throughput Relationship of the LAT Protocol}

\begin{figure*}[ht]
  \centering
  \subfigure[Power-throughput curves in terms of different $\chi^2$.]{
    \label{power_thput}
    \includegraphics[height = 2.65in]{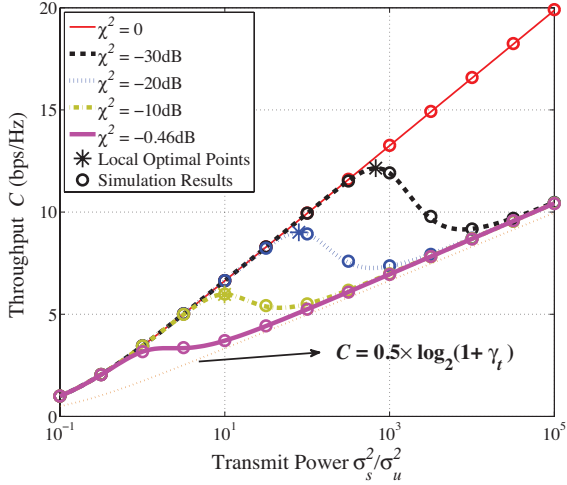}}
    \vspace{-0.5em}
  \hspace{0.6in}
  \subfigure[Existence of the local optimal transmit power.]{
    \label{checksolution}
    \includegraphics[height = 2.65in]{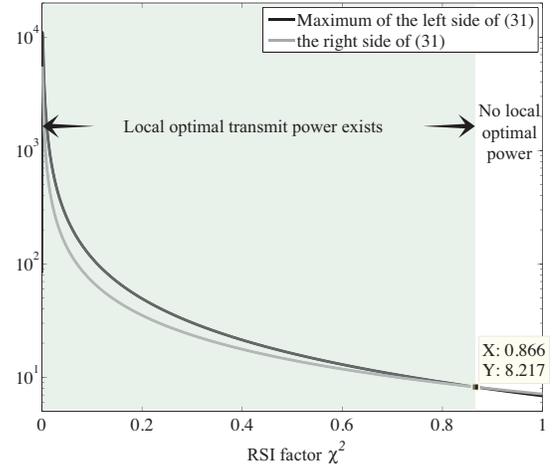}}
  \caption{Power-Throughput Curve in terms of different RSI factor $\chi^2$, where the probability of the PU's arrival $\mu = 1/500$, departure $\nu = 6/500$, the collision ratio $P_c = 0.1$, the sample number of a slot $N_s$ is 300, sensing SNR $\gamma_s = -5$dB.}
  \label{tradeoff} 
\end{figure*}

As is shown in Fig.~\ref{power_thput}, we consider the throughput performance of the LAT protocol in terms of secondary transmit power. The solid and dotted lines represent the analytical performance of the LAT protocol, and the asterisks~(*) denote the analytical local optimal transmit power. The small circles are the simulated results, which match the analytical performance well. The thin solid line depicts the ideal case with perfect RSI cancelation. Without RSI, the sensing performance is no longer affected by transmit power, and the throughput always goes up with the power. This line is also the upperbound of the LAT performance. The thick dash-dotted, dotted and dash lines in the middle are the typical cases, in which we can clearly observe the power-throughput tradeoff and identify the local optimal power, which is calculated from \eqref{simpdC=0}. With the decrease of RSI~($\chi^2$ from -10dB to -20dB to -30dB), the local optimal transmit power increases, and the corresponding throughput goes to a higher level. This makes sense since the smaller the RSI is, the better it approaches the ideal case, and the deterioration cause by self-interference becomes dominant under a higher power. According to Fig.~\ref{checksolution}, when $\chi^2$ is sufficiently large~ (0.85 in the figure), there exists no power-throughput tradeoff. We verify this result by the thick solid line denoting the cases when $\chi^2 = -0.46$dB $=0.9$. No local optimal point can be found in this curve, and the numerical results show that the differentiation is always positive.

One noticeable feature of Fig.~\ref{power_thput} is that when self-interference exists, all curves approach the thin dotted line $C = 0.5\log_2\left(1+\gamma_t\right)$ when the power goes up. This line indicates the case that the spectrum waste is 0.5. When the transmit power is too large, severe self-interference largely degrades the performance of spectrum sensing, and the false alarm probability becomes unbearably high. It is likely that whenever SU$_1$ begins transmission, the spectrum sensing result falsely indicates that the PU has arrived due to false alarm, and SU$_1$ stops transmission in the next slot. Once SU$_1$ becomes silent, it can clearly detect the PU's absence, and begins transmission in the next slot again. And the state of SU$_1$ changes every slot even when the PU does not arrive at all. In this case, the utility efficiency of the spectrum hole is approximately 0.5, which is clearly shown in Fig.~\ref{power_thput}. Also, it can be seen that the larger $\chi^2$ is, the earlier the sensing gets unbearable and the throughput approaches the orange line.

\vspace{-0.5em}
\subsection{Sensing Performance}

\begin{figure}[!t]
\centering
\includegraphics[width=3.6in]{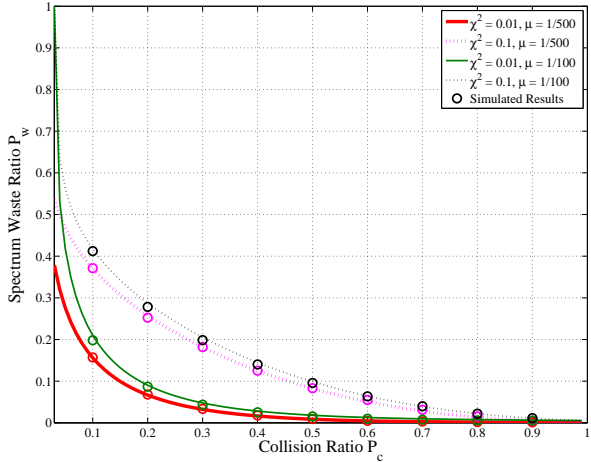}
\vspace{-1.5em}
\caption{ROCs in sensing. In this figure, the probability of the PU's arrival $\mu = 1/500$ and 1/100, departure $\nu = 6\mu$, the sample number of a slot $N_s$ is 300, normalized secondary transmit power $\sigma_s^2/\sigma_u^2 = 10$dB, sensing SNR $\gamma_s = -8$dB, and the RSI factor $\chi^2$ varies between 0.1 and 0.01.} \label{ROCcurve}
\vspace{-1em}
\end{figure}

In this subsection, we use the receiver operating characteristic curves~(ROCs) to present the sensing performance. In Fig.~\ref{ROCcurve}, with the sensing SNR $\gamma_s$ fixed on $-8$dB, we have the relation between the collision ratio and spectrum waste ratio. The thick lines denote the cases when the PU changes its state very slowly, while the fine lines represent the cases when the PU changes comparatively quickly. In Fig.~\ref{ROCcurve}, smaller area under a curve denotes better sensing performance. It can be seen that the thick lines are lower than the corresponding fine lines, which indicates that when the PU changes its state slowly, the spectrum holes can be utilized with higher efficiency. This is because the spectrum waste due to the state change, i.e., re-access and departure of the PU happens less frequently. Also, comparing the solid and dotted lines with the same $\mu$, it can be found that smaller RSI leads to better sensing performance, and the impact of the RSI can be significant. It is noteworthy that the ratio of spectrum waste of the LAT can be quite close to zero if the self-interference can be effectively suppressed, and the constraint of collision ratio is not too strict. However, recall the conventional listen-before-talk, the spectrum waste ratio can theoretically never be suppressed lower than the sensing time ratio in a slot.

\vspace{-0.5em}
\subsection{Impact of the RSI Factor $\chi^2$}
\begin{figure}[!t]
\centering
\includegraphics[width=3.6in]{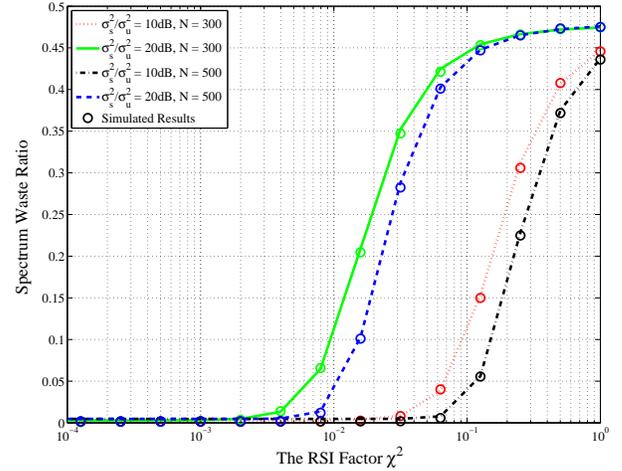}
\vspace{-1.5em}
\caption{Secondary throughput versus the RSI factor $\chi^2$, in which the collision ratio is 0.1, sensing SNR $\gamma_s = -5$dB, the normalized secondary transmit power $\sigma_s^2/\sigma_u^2$ varies between 10dB to 20dB, and the numbers of samples in a slot varies between 300 and 500, with the probability of the PU's arrival per slot $\mu$ varies between 1/500 and 1/300, departure probability per slot $\nu$ varies between 6/500 and 6/300.} \label{X_Pw}
\vspace{-1em}
\end{figure}

In Fig.~\ref{X_Pw}, we consider the impact of the RSI factor on the sensing performance. We fix the constraint of $P_c$ as 0.1, and evaluate the spectrum waste ratio under various $\chi^2$. It can be seen from Fig.~\ref{X_Pw} that with the increase of $\chi^2$, the spectrum waste ratio increases from zero to approximately 0.5. This is reasonable in the sense that with sufficiently small RSI factor, the RSI can be neglected compared with PU's signal and noise, and SUs can fully utilize the spectrum holes. When the RSI factor is moderate or close to 1, which indicates that the RSI cannot be suppressed well, the secondary signal may overwhelm the PU's signal, leading to unreliability of sensing, and the SUs are likely to stop communication due to false alarm. Note that the asymptotic value of the spectrum waste ratio when the RSI is large is 0.5, which is in accordance with the results in Fig.~\ref{power_thput}.

Besides, it can be seen that when the normalized power of RSI~($\chi^2\sigma_s^2/\sigma_u^2$) ranges from approximately [0.1, 10], the spectrum waste ratio changes fast, and when the normalized power of RSI is below 0.1, the waste ratio remains at a low level. This feature can be utilized to design the protocol parameters to achieve full utilization of the spectrum holes. Also, when the PU's state change rate remains unchanged while the slot length enlarges, it can be seen that the sensing performance becomes better, especially at the points when the normalized power of RSI ranges from approximately [0.1, 10], where more samples in a slot would help improve sensing performance significantly.

\section{Conclusions}%

In this paper, we proposed a LAT protocol that allows SUs to simultaneously sense and access the spectrum holes. Taken the impact of the residual self-interference on sensing performance into consideration, we designed an adaptively-changed sensing threshold for energy detection. Spectrum utilization efficiency and secondary throughput under the LAT protocol has been provided in closed-form, based on which a unique tradeoff between the secondary transmit power and the secondary throughput has been reported, i.e., the increase of transmit power does not always yield the improvement of SU's throughput, and a mediate value is required to achieve the local optimal performance. Simulation results have verified the existence of the power-throughput tradeoff, and shown that the SUs can efficiently utilize the spectrum holes under the LAT protocol.

The proposed LAT protocol has the potential to allow the FD SUs to fully utilize the spectrum holes, given that the SUs no longer need to periodically suspend their transmission for sensing, and can react promptly to the spectrum opportunity. Besides the basic model considered in this paper, the LAT protocol can be readily extended to many other CR scenarios, like the multi-user and multi-channel cases. With simultaneous sensing and transmission, the collision between multiple SUs is likely to be shorten, and the performance of the whole secondary network is likely to enjoy a significant improvement.

\appendices

\section{Derivation of Table.~\ref{FDGaussian}}

We first provide the general properties of the test statistics. Given that each $y(n)$ in \eqref{test} is i.i.d., the mean and the variance of $M$ can be calculated as
\begin{equation}
\mathbb{E}\left[M\right] = \mathbb{E}\left[\left|y\right|^2\right];~{\rm{var}}\left[M\right] = \frac{1}{N_s}{\rm{var}}\left[\left|y\right|^2\right].\nonumber
\end{equation}
Further, if the received signal $y$ is complex-valued Gaussian with mean zero and variance $\sigma_y^2$, we have
\begin{equation}
\mathbb{E}\left[M\right] = \sigma_y^2,\nonumber
\end{equation}
and
\begin{equation}\label{CSCG}
{\rm{var}}\left[M\right] = \frac{1}{N_s}\left(\mathbb{E}\left[\left|y\right|^4\right]-\sigma_y^4\right)=\frac{\sigma_y^4}{N_s}.
\end{equation}

Then we consider the concrete form of the received signal under each hypothesis. In the LAT protocol, given the PU signal, RSI, and i.i.d. noise, the received signal $y$ is complex-valued Gaussian with zero mean. The variance of $y$ under the four hypothesises are as follow:
\begin{equation}
\sigma_y^2 = \left\{ {\begin{array}{*{20}{c}}
&{\left(1+\gamma_s\right)\sigma _u^2}~~~~&\mathcal{H}_{01},\\
&{\left(1+\gamma_i\right)\sigma _u^2}~~~~&\mathcal{H}_{10},\\
&{\sigma _u^2}~~~~&\mathcal{H}_{00},\\
&{\left(1+\gamma_s + \gamma_i\right)\sigma _u^2}~~~~&\mathcal{H}_{11}.
\end{array}} \right.
\end{equation}
By substituting them into \eqref{CSCG}, we can obtain the results in Table~\ref{FDGaussian}.

\section{Derivation of the Optimal Transmit Power}
The optimal power $\widehat{\sigma_s^2}$ satisfies
\begin{equation}\label{dC}
{\left. {\frac{{dC}}{{d\sigma _s^2}}} \right|_{\widehat {\sigma _s^2}}} = 0.
\end{equation}

The differentiation of the secondary throughput can be derived as shown in \eqref{dCorigin} at the top of next page, and with $\frac{dC}{d\sigma_s^2} = 0$, equation \eqref{dc=0_long} can be obtained,
\begin{figure*}
\begin{equation}\label{dCorigin}
\begin{split}
\frac{dC}{d\sigma_s^2} =& - {\log _2}\left( {{\gamma _t} + 1} \right) \cdot \frac{{\exp \left( { - \frac{{{\rho ^2}}}{2}} \right) \cdot \left( {\frac{1}{\mu } - 1} \right) \cdot \frac{{{\gamma _s}{\chi ^2}\left( {{\mathcal{Q}^{ - 1}}\left( {1 - {P_m}} \right) + \sqrt {{N_s}} } \right)}}{{{{\left( {{\gamma _i} + 1} \right)}^2}}}}}{{\sqrt {2\pi } {{\left[ {\left( {\frac{1}{\mu } - 1} \right) \cdot \left( {\mathcal{Q}\left( \rho  \right) - P_f^0 + 1} \right) + 1} \right]}^2}}} \cdot \left( {\left( {\frac{1}{\mu } - 1} \right) \cdot \left( {1 - P_f^0} \right) + {P_m}} \right)\\
 &- \frac{{\sigma _t^2}}{{\ln 2 \cdot \left( {{\gamma _t} + 1} \right)}} \cdot \left( {\frac{\mu }{2} + \frac{{\mathcal{Q}\left( \rho  \right) \cdot \left( {\frac{1}{\mu } - 1} \right) - {P_m} + 1}}{{\left( {\frac{1}{\mu } - 1} \right) \cdot \left( {\mathcal{Q}\left( \rho  \right) - P_f^0 + 1} \right) + 1}} - 1} \right),
\end{split}
\end{equation}
where $\rho = \mathcal{Q}^{-1}\left(1-P_m\right)\cdot\left(\frac{\gamma_s}{\gamma_i+1}+1\right)+\frac{\gamma_s}{\gamma_i+1}\sqrt{N_s}$, i.e., $\mathcal{Q}\left(\rho\right) = P_f^1$.
\end{figure*}
\begin{figure*}[ht]
\begin{equation}\label{dc=0_long}
\begin{split}
&\ln \left( {{\gamma _t} + 1} \right) \cdot \frac{{\exp \left( { - \frac{{{\rho ^2}}}{2}} \right) \cdot \left( {\frac{1}{\mu } - 1} \right) \cdot \frac{{{\gamma _s}{\chi ^2}\left( {{\mathcal{Q}^{ - 1}}\left( {1 - {P_m}} \right) + \sqrt {{N_s}} } \right)}}{{\left( {{\gamma _i} + 1} \right)^2}}}}{{\sqrt {2\pi } {{\left[ {\left( {\frac{1}{\mu } - 1} \right) \left(\mathcal{Q}\left(\rho\right)-P_f^0+1\right)  + 1} \right]}^2}}} \cdot \left( {\left( {\frac{1}{\mu } - 1} \right) \cdot \left( {1 - P_f^0} \right) + {P_m}} \right)\\
&+ \frac{\sigma _t^2}{\gamma_t + 1} \cdot \left( {\frac{\mu }{2} + \frac{{\left( {\frac{1}{\mu } - 1} \right) \cdot \mathcal{Q}\left( \rho  \right) - {P_m} + 1}}{{\left( {\frac{1}{\mu } - 1} \right) \left(\mathcal{Q}\left(\rho\right)-P_f^0+1\right)  + 1}} - 1} \right) = 0,
\end{split}
\end{equation}
\hrulefill
\end{figure*}
which can be simplified as
\begin{equation}\small\label{dC=0}
\ln \left( {{\gamma _t} + 1} \right) \frac{{\exp \left( { - \frac{{{\rho ^2}}}{2}} \right) \left( {\frac{1}{\mu } - 1} \right) \Xi}}{{\sqrt {2\pi } \left(\gamma_i+1\right)^2\alpha }}\kappa  + \frac{\sigma _t^2}{\gamma_t+1} \left( {\frac{\mu }{2} - \kappa } \right) = 0.
\end{equation}

When $\mu$ is sufficiently small, the notations can be simplified as
\begin{equation}
\begin{split}
&\alpha  = \frac{1}{\mu } \cdot \left( {\mathcal{Q}\left( \rho  \right) - P_f^0 + 1} \right),\\
&\kappa  = \frac{{  {1 - P_f^0} }}{{ \mathcal{Q}\left(\rho\right)-P_f^0+1}},
\end{split}
\end{equation}
and \eqref{dC=0} becomes
\begin{equation}
\frac{\ln \left( {{\gamma _t} + 1} \right)}{\left(\gamma_i + 1\right)^2} \cdot \frac{{\exp \left( { - \frac{{{\rho ^2}}}{2}} \right) \cdot \Xi}}{{\sqrt {2\pi } \left( {\mathcal{Q}\left( \rho  \right) - P_f^0 + 1} \right)}} \cdot \kappa  - \frac{\sigma _t^2\kappa}{\gamma_t + 1}  = 0,
\end{equation}
i.e.,
\begin{equation}
\exp \left( { - \frac{{{\rho ^2}}}{2}} \right)\frac{\left(\gamma_t + 1\right)\ln\left( {{\gamma _t} + 1} \right)}{{{\left( {{\gamma _i} + 1} \right)^2}}} = \frac{{\sqrt {2\pi } \sigma _t^2}}{\Xi }\left( {1 - P_f^0 + \mathcal{Q}\left( \rho  \right)} \right).
\end{equation}

\end{document}